\documentclass[11pt,twoside]{article}
\usepackage[margin=1in]{geometry}

\usepackage{url,hyperref,xspace}
\usepackage[usenames]{color}
\usepackage{times}
\usepackage{verbatim,makeidx}
\usepackage{algorithm}
\usepackage{algorithmic}
\usepackage{amsmath}
\usepackage{amssymb}
\usepackage{amsthm}
\usepackage{dsfont}
\usepackage{color}
\usepackage{colortbl}
\usepackage{float}
\usepackage{cite}
\usepackage{pbox}
\usepackage{bbm}
\usepackage{authblk}
\usepackage{graphicx}
\usepackage{caption}
\usepackage{subcaption}

\allowdisplaybreaks

\newcounter{ALC@tempcntr}

\theoremstyle{plain}
\newtheorem{theorem}{Theorem}
\newtheorem{proposition}{Proposition}
\newtheorem{definition}{Definition}
\newtheorem{remark}{Remark}

\theoremstyle{definition}
\newtheorem{example}{Example}

\theoremstyle{remark}

\newcommand{\beq}{\begin{eqnarray}}
\newcommand{\eeq}{\end{eqnarray}}

\newcommand{\field}[1]{\mathbb{#1}}

\newcommand{\F}{\field{F}}
\newcommand{\B}{\field{B}}
\newcommand{\Lf}{\field{L}}

\newfont{\bbb}{msbm10 scaled 500}

\newfont{\bb}{msbm10 scaled 1100}

\newcommand{\cv}{{\bf c}}

\newcommand{\fv}{{\bf f}}

\newcommand{\xv}{{\bf x}}
\newcommand{\yv}{{\bf y}}

\newcommand{\Em}{{\bf E}}

\newcommand{\Hm}{{\bf H}}

\newcommand{\Pm}{{\bf P}}

\newcommand{\Cc}{{\cal C}}

\newcommand{\Kc}{{\cal K}}

\newcommand{\Rc}{{\cal R}}
\newcommand{\Sc}{{\cal S}}

\newcommand{\remove}[1]{}

\theoremstyle{definition}

\theoremstyle{remark}

\newcommand{\latexe}{{\LaTeX\kern.125em2%
                      \lower.5ex\hbox{$\varepsilon$}}}

\chardef\bslash=`\\	

\makeatletter		
\def\square{\RIfM@\bgroup\else$\bgroup\aftergroup$\fi
\vcenter{\hrule\hbox{\vrule\@height.6em\kern.6em\vrule}
\hrule}\egroup}\makeatother\makeindex

\newcommand\pr[1]{{\mathbb P}\big\{#1\big\}}

\definecolor{OXO-emph}{RGB}{153,0,0}

\newcommand\ceilb[1]{\left\lceil #1 \right\rceil}
\newcommand\ceilbb[1]{\bigl\lceil #1 \bigr\rceil}
\newcommand\floorb[1]{\left\lfloor #1 \right\rfloor}
\newcommand\floorbb[1]{\bigl\lfloor #1 \bigr\rfloor}

\DeclareMathAlphabet{\mathpzc}{OT1}{pzc}{m}{it}
\newcolumntype{?}{!{\vrule width 1.5pt}}

\title{New MDS codes with small sub-packetization and near-optimal repair bandwidth}

\author{Venkatesan Guruswami and Ankit Singh Rawat}%
\affil{Computer Science Department, \\ Carnegie Mellon University, \\Pittsburgh, 15213.\\ {E-mail:~venkatg@cs.cmu.edu, asrawat@andrew.cmu.edu}}

\begin{document}

\maketitle



\begin{abstract} 
An $(n, M)$ vector code $\mathcal{C} \subseteq \mathbb{F}^n$ is a collection
of $M$ codewords where $n$ elements (from the field $\mathbb{F}$) in each of
the codewords are referred to as code blocks. Assuming that $\mathbb{F} \cong
\mathbb{B}^{\ell}$, the code blocks are treated as $\ell$-length vectors over
the base field $\mathbb{B}$. Equivalently, the code is said to have the
sub-packetization level $\ell$. This paper addresses the problem of constructing MDS vector codes which enable exact reconstruction of each code
block by downloading small amount of information from the remaining code
blocks. The repair bandwidth of a code measures the information flow from the
remaining code blocks during the reconstruction of a single code block. This
problem naturally arises in the context of distributed storage systems as the
node repair problem~\cite{dimakis}. Assuming that $M = |\mathbb{B}|^{k\ell}$, the repair
bandwidth of an MDS vector code is lower bounded by $\big(\frac{n - 1}{n -
k}\big)\cdot \ell$ symbols (over the base field $\mathbb{B}$) which is also
referred to as the cut-set bound~\cite{dimakis}. For all values of $n$ and $k$, the MDS
vector codes that attain the cut-set bound with the sub-packetization level
$\ell = (n-k)^{\lceil{{n}/{(n-k)}}\rceil}$ are known in the literature~\cite{SAK15, YeB16b}.

This paper presents a construction for MDS vector codes which simultaneously
ensures both small repair bandwidth and small sub-packetization level. The
obtained codes have the smallest possible sub-packetization level $\ell = O(n -
k)$ for an MDS vector code and the repair bandwidth which is at most twice the
cut-set bound. The paper then generalizes this code construction so that the
repair bandwidth of the obtained codes approach the cut-set bound at the cost
of increased sub-packetization level. The constructions presented in this paper
give MDS vector codes which are linear over the base field $\mathbb{B}$.
\end{abstract}

\newpage


\section{Introduction}
\label{sec:intro}

Maximum distance separable (MDS) codes are considered to be an attractive solution for information storage as they operate at the optimal storage vs. reliability trade-off given by the Singleton bound~\cite{MacSlo}. For a given amount of information to be stored and available storage space, the MDS codes can tolerate the maximum number of worst case failures without losing the stored information. However, the applicability of the MDS codes in modern storage systems also depends on their ability to efficiently regenerate parts of a codeword from the rest of the codeword. Consider a distributed storage system which employs an MDS code to store information over a network of storage nodes such that each storage node stores a part of a codeword from the MDS code. Exact regeneration (repair) of the content stored in a node with the help of the content stored in the remaining nodes is useful to reinstate the system in the event of a permanent node failure. Similarly, this also enables access to the information stored on a temporarily unavailable node with the help of available nodes in the system. Therefore, among all MDS codes, the ones with more efficient exact repair mechanisms are preferred for deployment in modern distributed storage systems.

In \cite{dimakis}, Dimakis et al. study the repair problem in distributed storage systems and introduce {\em repair bandwidth}, the amount of data downloaded during a node repair, as a metric to compare various codes in terms of the efficiency of their node repair mechanisms. Let $\Cc \subseteq \F^n$ be an MDS code with $|\Cc| = |\F|^k$ codewords, each of length $n$ (over $\F$). Given a file $\fv \in \F^k$, it is mapped to a codeword in $\Cc$. Subsequently, the $n$ symbols (over $\F$) in the associated codeword are stored in $n$ distinct storage nodes in the system. For an MDS code, it is straightforward to achieve a repair bandwidth of $k$ symbols (over $\F$) by contacting any $k$ remaining nodes and downloading the $k$ distinct symbols stored on these nodes. This follows from the fact that any $k$ symbols of a codeword from an MDS codes are sufficient to reconstruct the {\em entire} codeword.
Note that the repair bandwidth of $k$ symbols (over $\F$) is the best possible if we are allowed to contact only $k$ remaining storage nodes during the repair process. Furthermore, it is not possible to regenerate a code symbol by contacting less than $k$ remaining code symbols of a codeword in an MDS code. This motivates Dimakis et al. to look for potentially lowering the repair bandwidth for repair of a single node by contacting $d \geq k$ remaining nodes in the system and downloading partial data stored on each of the contacted nodes. 

Assuming that the MDS code $\Cc$ is defined over the field $\F \cong \B^{\ell}$, we can view each of the $n$ symbols (over $\F$) in a codeword as an $\ell$-length vector over the base field $\B$. Given this vector representation of the MDS code, the repair bandwidth of an MDS code is lower bounded by~\cite{dimakis,HLKB15} 
\begin{align}
\label{eq:cut_set_d}
\Big(\frac{d}{d - k + 1}\Big)\cdot \ell~~~\text{symbols (over $\B$)}.
\end{align}
In the particular case, when $d = n -1$, i.e., all the remaining nodes in the system are contacted during the repair process, the bound on the repair bandwidth reduces to
$
\Big(\frac{n - 1}{n - k}\Big)\cdot \ell~~~\text{symbols (over $\B$)}.
$

The bound in \eqref{eq:cut_set_d} is referred to as the cut-set bound in the literature. The problem of constructing MDS codes with optimal repair-bandwidth (cf.~\eqref{eq:cut_set_d}) has been explored in \cite{RSK11, CJMRS13, zigzag13, PapDimCad13, SAK15, RKV16a, GoparajuFV16, YeB16b} and references therein. Note that as the number of the nodes contacted during the repair process $d$ gets larger, the optimal repair bandwidth defined by the cut-set bound becomes significantly smaller than the naive repair bandwidth of $k$ symbols (over $\F$) or $k\ell$ symbols (over $\B$).

This paper explores a trade-off between the sub-packetization level $\ell$ and the repair bandwidth for the MDS codes. The MDS codes that work with small sub-packetization level in addition to having small repair bandwidth are of great practical importance in distributed storage systems. The smaller sub-packetization leads to easier system implementation as it provides the system designer with greater flexibility in terms selecting various system parameters. As an example, one does not have to combine the data from multiple different sources to meet the larger sub-packetization requirement in order to be able to enable efficient repair mechanism. Note that for the given system parameters $n, k$ and $d$, we require $k \ell$ symbols (over $\B$) worth of data to store (using an MDS code with sub-packetization level $\ell$) so that we utilize the storage space in the most efficient manner. As an illustration of another practical advantage of having smaller sub-packetization level, consider a scenario where a MDS code requires a large sub-packetization level, e.g., say $\ell \geq 2^n$. This implies that using storage nodes (disks) with storage capacity of $\ell$ symbols (over $\B$), one can only design a storage system with at most $\log_2 \ell$ nodes. Therefore, larger sub-packetization level can lead to a reduced design space in terms of various system parameters.

\noindent \textbf{Our contributions.}~We present a new and simple construction for MDS codes which have small sub-packetization level while allowing for exact repair of all code symbols with near-optimal repair bandwidth. This construction highlights a trade-off between the sub-packetization level and the repair bandwidth for exact repair. The construction is obtained by utilizing the parity-check view of a linear code. Assuming that the desired sub-packetization level is $\ell$, we start with a parity-check matrix of a simple MDS code (over $\F \cong \B^{\ell}$) of length $n$ which is obtained by stacking $\ell$ codewords from $\ell$ independent MDS codes (over $\B$) of length $n$. We then carefully replace some of the zero entries of this parity-check matrix with non-zeros elements from $\B$ and obtain a parity check matrix of a new MDS code that has an exact repair mechanism with small repair bandwidth. It follows from the construction that the obtained MDS codes are linear over the field $\B$. 

We note that throughout this paper we consider the setting with $d = n-1$, i.e., all the remaining code blocks contribute to the exact repair of a single code block. We list the parameters and some of the exact-repair related properties of the obtained codes in the following.

\begin{itemize}
\item \textbf{Codes with repair bandwidth at most twice the cut-set bound.}~We first present a family of MDS codes that have sub-packetization level $\ell = n - k$ and the repair bandwidth that is strictly less than $2(n-1)$ symbols (over $\B$). Note that this is twice the cut-set bound (cf.~\eqref{eq:cut_set_d}) which takes the value $(n-1)$ symbols (over $\B$) for $\ell = n - k$. {We argue that the sub-packetization level $\ell = \Omega(n - k)$ is the smallest that we can hope for an MDS code with the aforementioned guarantee on its repair bandwidth (See Appendix~\ref{appen:appen_sub}).}
\item \textbf{Codes with repair bandwidth approaching the cut-set bound.}~We generalize the ideas used in the construction with $\ell = n - k$ to obtain the MDS codes that have improved repair bandwidth at the cost of increased sub-packetization level. In particular, for an integer $t \geq 2$, we obtain a family of MDS codes with sub-packetization level $\ell = (n - k)^t$ and the repair bandwidth which is at most $(1 + \frac{1}{t})$ times the value of the cut-set bound.
\item \textbf{Exact repair using repair-by-transfer schemes}~The codes presented in this paper are MDS codes defined over $\F = \B^{\ell}$ which are linear over the base field $\B$. These codes are referred to as linear MDS vector codes or linear MDS vector codes in the literature. For these codes, $n$ code blocks are stored in the form of an $\ell$-length vectors (over $\B$) in $n$ distinct nodes. The exact repair of each code block in these codes involves downloading a subset of the symbols from the remaining code blocks. Such repair mechanisms are referred to as the {\em uncoded repair} or {\em repair-by-transfer} in the literature. The repair-by-transfer schemes form a sub-class of all possible linear repair schemes where a contacted node can potentially send symbols (over $\B$) which are linear combinations of all $\ell$ symbols of the code-block stored on this node. We note that repair-by-transfer is desirable over other complicated repair schemes due to its operational simplicity and the minimal computation requirements at the contacted nodes.
\end{itemize}

\noindent \textbf{Organization.}
We introduce the necessary background along with a discussion on the related work in Section~\ref{sec:background}. In Section~\ref{sec:main} we define the notion of near-optimal exact-repairable MDS codes and summarize the code parameters achievable by our construction. We present the construction of the MDS codes with $\ell = n - k$ and repair bandwidth at most twice the value of the cut-set bound in Section~\ref{sec:twiceRB}. In Section~\ref{sec:RB_t} we present the the general construction that gives MDS codes with their repair bandwidth approaching the optimal repair bandwidth. We conclude the paper in Section~\ref{sec:conclusion} where we comment on the constructions of the codes with general values of $d$ (the number of blocks contributing to the repair process) and discuss other directions for future work.

\section{Background and related work}
\label{sec:background}

In this section we formally introduce vector codes and the related concepts used in this paper. 
We then describe the exact repair problem in the context of distributed storage systems and survey the related work. 


\subsection{Preliminaries}
\label{sec:prelims}

Let $\mathbbm{1}_{\{\cdot\}}$ denote the standard indicator function. Given two vectors $\xv, \yv \in \B^{n\ell}$, we defined the (block) Hamming distance between them as  $d_{\rm H}(\xv, \yv) = \sum_{i = 1}^{n}\mathbbm{1}_{\{\xv_i \neq \yv_i \}},$
where for $i \in [n]$, we have $\xv_i = (x_{(i-1)\ell + 1},\ldots,x_{i\ell})$ and $\yv_i = (y_{(i-1)\ell + 1},\ldots,y_{i\ell})$. For a finite field $\B$, we say that a set of vectors $\Cc \subseteq \B^{n\ell}$ is an $(n, M, d_{\min}, \ell)_{|\B|}$ {\em vector code} (or in short, $(n, M)$ vector code) if we have $|\Cc| = M$ and $d_{\min} := \min_{\xv \neq \yv \in \Cc}d_{\rm H}(\xv, \yv)$.  Given a codeword $\cv = (c_1, c_2,\ldots, c_{n\ell}) \in \Cc \subseteq \B^{n\ell}$, we use $$\cv_i = (c_{(i - 1)\ell + 1}, c_{(i - 1)\ell + 2},\ldots, c_{i\ell}) \in \B^{\ell}$$ to denote the $i$-th 
{\em code block} in the codeword. When the code $\Cc$ spans a $\B$-linear subspace of dimension $K = \log_{|\B|}{M}$, we call $\Cc$ to be a {\em linear} vector code and refer to it as an $[n, K, d_{\min}, \ell]_{|\B|}$ vector code. An $[n, K = \log_{|\B|}{M}, d_{\min}, \ell]_{|\B|}$ vector code is said to be a {\em linear MDS vector code} if we have
$
K = \log_{|\B|}{M} = k\ell~~\text{and}~~d_{\min} = n - k + 1.
$
Note that an $[n, k\ell, d_{\min}, \ell]_{|\B|}$ vector code can be defined by an $(n-k)\ell \times n\ell$ parity-check matrix
\begin{align}
\label{eq:ArrayParity}
\Hm = \left( \begin{array}{cccc}
H_{1,1} & H_{1, 2} & \cdots & H_{1,n}\\
H_{2,1} & H_{2, 2} & \cdots & H_{2,n}\\
\vdots & \vdots & \ddots & \vdots \\
H_{r,1} & H_{r, 2} & \cdots & H_{r,n}\\
\end{array} \right) \in \B^{(n-k)\ell \times n\ell},
\end{align}
where each $H_{i, j}$ is an $\ell \times \ell$ matrix with its entries belonging to the finite field $\B$. For a set $\Sc = \{i_1, i_2,\ldots, i_{|\Sc|}\} \subseteq [n]$, we define the $(n-k) \ell \times |\Sc|\ell$ matrix $\Hm_{\Sc}$ as follows. 
\begin{align}
\label{eq:ParitySub}
\Hm_{\Sc} = \left( \begin{array}{cccc}
H_{1, i_1} & H_{1, i_2} & \cdots & H_{1, i_{|\Sc|}}\\
H_{2, i_1} & H_{2, i_2} & \cdots & H_{2, i_{|\Sc|}}\\
\vdots & \vdots & \ddots & \vdots \\
H_{r, i_1} & H_{r, i_2} & \cdots & H_{r, i_{|\Sc|}}
\end{array} \right) \in \B^{(n-k) \ell \times |\Sc|\ell}.
\end{align}
Note that the matrix $\Hm_{\Sc}$ comprises those coefficients in the linear constraints defined by the parity-check matrix $\Hm$ that are associated with the code blocks indexed by the set $\Sc \subseteq [n]$. The parity-check matrix $\Hm$ defines an MDS vector codes if for every $\Sc \subseteq [n]$ with $|\Sc| = n-k$, the $r\ell \times r\ell$ sub-matrix $\Hm_{\Sc}$ is full rank.

\subsection{Exact-repair problem for MDS vector codes}
\label{sec:exact-repair}

Let $\Cc \in \B^{n \ell}$ be a vector code with $K = \log_{|\B|}{|\Cc|} = \log_{|\B|}{M}$. Consider an encoding process which encodes a file $\fv \in \B^{K}$ to a codeword $\cv = (\cv_1, \cv_2,\ldots, \cv_n) \in \Cc \subseteq \B^{n\ell}$, where for every $i \in [n]$ we have $\cv_i \in \B^{\ell}$. We require the encoding process to ensure that the original file $\fv$ can be reconstructed from any $k$ out of the $n$ code blocks in the codeword $\cv$, i.e., for any $\Kc \subseteq [n]$ such that $|\Kc| = k$, $\fv$ can be reconstructed from the code blocks $\{\cv_i\}_{i \in \Kc}$. The exact-repair problem imposes the requirement that for every $i \in [n]$ and $\Rc \subseteq [n]\backslash \{i\}$ with $|\Rc| = d$, we have a collection of functions, 
$
\big\{h^{(i)}_{j, \Rc}~:\B^{\ell} \rightarrow \B^{\beta_{j,i}}\big\}_{j \in \Rc}
$
such that $\cv_i$ is a function of the symbols in the set $\{h^{(i)}_{j, \Rc}(\cv_j)\}_{j \in \Rc}$. This implies that  for every $i \in [n]$, the code block $\cv_i$ can be {\em exactly repaired} (regenerated) by contacting any $d$ out of $n-1$ remaining code blocks in the codeword $\cv$ (say indexed by the set $\Rc \subseteq [n]\backslash \{i\}$) and downloading at most $\sum_{j \in \Rc}\beta_{j, i}$ symbols (over $\B$) from the contacted code blocks.

In \cite{dimakis}, Dimakis et al. formally study the repair problem for vector codes in the setup described above. They introduce {\em repair bandwidth}, the total number of symbols downloaded during the repair process, as a measure to characterize the efficiency of the repair process\footnote{Dimakis et al. consider a broader repair framework, namely {\em functional repair} framework~\cite{dimakis}. Under functional repair framework, $\tilde{\cv}_i \in \B^{\ell}$ which may potentially be different from the code block under repair $\cv_i \in \B^{\ell}$ is an acceptable outcome of the repair process as long as it preserves certain properties of the original codeword. For further details, we refer the reader to \cite{dimakis, DRWS2011}. Here, we note that the lower bounds obtained for the functional repair problem are also applicable to the exact repair problem considered in this paper.}. Assuming that we download the same number of symbols from each of the contacted code blocks, i.e., $\beta_{j, i} = \beta$ (symbols over $\B$) for all $j \in \Rc$, Dimakis et al. obtain the following cut-set bound on the repair bandwidth of an MDS vector code~\cite{dimakis}.
\begin{align}
\label{eq:cut-set}
d\beta \geq \left(\frac{d}{d - k + 1}\right)\cdot \ell~~\text{(symbols over $\B$)}.
\end{align}

Interestingly, the lower bound on the repair bandwidth of an MDS vector code given in \eqref{eq:cut-set} continues to hold even when the contacted nodes contribute different number of symbols during the repair process~\cite{HLKB15}, i.e., for every $i \in [n]$, we have
\begin{align}
\label{eq:cut-set-gen}
\sum_{j \in \Rc}\beta_{j, i} \geq \left(\frac{d}{d - k + 1}\right)\cdot \ell~~\text{(symbols over $\B$)},~~\forall~\Rc \subseteq [n]\backslash \{i\}~\text{s.t.}~|\Rc| = d.
\end{align}

The problem of constructing {\em exact-repairable MDS vector code}, MDS vector codes that enable exact repair of all code blocks, with optimal repair bandwidth (cf.~\eqref{eq:cut-set}) has been explored by many researchers. In \cite{RSK11}, Rashmi et al. present an explicit construction for exact-repairable MDS vector codes. This construction works with the sub-packetization level $\ell = d - k + 1$. However, the construction requires $2k - 2 \leq d \leq n - 1$, which leads to low information rate, i.e., $\frac{k}{n} \leq \frac{1}{2} + \frac{1}{2n}$. Towards constructing high-rate exact-repairable MDS codes with optimal repair-bandwidth, Cadambe et al.~\cite{CJMRS13} show the existence of such codes when sub-packetization level approaches infinity. Motivated by this result, the problem of designing high-rate exact-repairable MDS codes with finite sub-packetization level and optimal repair bandwidth is explored in \cite{PapDimCad13, zigzag13, SAK15, WTB12, Cadambe_poly, RKV16a, GoparajuFV16, YeB16a, YeB16b} and references therein.

The code construction based on Hadamard matrices from \cite{PapDimCad13} requires $d = n-1$ and $n - k = 2$. In \cite{zigzag13}, Tamo et al. propose the zigzag code construction for $d = n-1$ and every value of $r = n - k$. This construction enables exact repair of only $k$ (systematic) code blocks. Wang et al.~\cite{zigzag_allerton11} generalize the zigzag code construction to enable exact repair of all code symbols. The constructions presented in \cite{PapDimCad13, zigzag13, zigzag_allerton11}  work with the sub-packetization level $\ell$ which is exponential in $k$. For $d = n-1$ and all values of $r = n - k$, Sasidharan et al.~\cite{SAK15} construct exact-repairable MDS vector codes that have optimal repair bandwidth and work with the sub-packetization level  $\ell = (n - k)^{\ceilb{\frac{n}{n-k}}}$. Note that for $r = n - k = \Omega(n)$, this construction provides the codes with the sub-packetization level which is polynomial in $n$.  {The construction with $d = n - 1$ and the similar sub-packetization levels that enable exact-repair of only $k$ systematic nodes are also presented in \cite{WTB12, Cadambe_poly}.} The construction from \cite{SAK15} is generalized to work for all possible values of $k \leq d \leq n - 1$ with the sub-packetization level $\ell = (d - k + 1)^{\ceilb{\frac{n}{d - k + 1}}}$ in \cite{RKV16a}. 

The MDS codes presented in \cite{PapDimCad13, zigzag13, zigzag_allerton11, WTB12, Cadambe_poly} are obtained by designing a suitable generator matrix for these codes. On the other hand, \cite{SAK15, RKV16a} design the proposed codes by constructing parity check matrices with certain combinatorial structures. We note that in most of these constructions, certain elements in the generator/parity-check matrices are not explicitly specified. These papers argue the existence of good choices for these elements provided that the field size is large enough. Recently, Ye and Barg~\cite{YeB16b} have presented a fully explicit construction for MDS codes with $d = n-1$ and the sub-packetization level $\ell = (n - k)^{\ceilb{\frac{n}{n - k}}}$ by designing the associated parity-check matrices. This construction is closely related to the construction presented in \cite{SAK15} in terms of the combinatorial structure of the parity-check matrix. We also note that the construction from \cite{YeB16b} also works for general values of $k \leq d \leq n - 1$ with suitably modified sub-packetization levels similar to the sub-packetization levels used in \cite{RKV16a}. In Table~\ref{tab:comparison} we summarize code parameters of state-of-the-art constructions in different settings.

\bgroup
\def\arraystretch{1.4}
\begin{table}[t!]
\footnotesize
\centering
\captionsetup{justification=centering}
\caption{Comparison of the proposed code construction with various existing code constructions for MDS codes that have small repair bandwidth for exact repair. We focus only on the setting with $d = n - 1$.}
\begin{tabular}{c|c|c|c|c|}
  \hline \hline
  Code construction & Sub-packetization level & Repair bandwidth & Repair by transfer & Information rate \\
  \hline
  Rashmi et al., 2011 \cite{RSK11} & $\ell = n - k$ & $\big(\frac{n-1}{n-k}\big)\cdot \ell$ & No & $2 \leq 2k \leq n + 1$ \\
  \hline
  Ye and Barg, 2016 \cite{YeB16b} & $\ell = (n - k)^{\ceilb{\frac{n}{n-k}}}$ & $\big(\frac{n-1}{n-k}\big)\cdot \ell$ & Yes & $1 \leq k \leq n - 1$ \\
  \hline
  \pbox{20cm}{~~~~~~~~~~This paper \\ (design parameter $t \geq 1$)} & $\ell = (n - k)^{t}$ & $\leq (1 + \frac{1}{t})\cdot\big(\frac{n-1}{n-k}\big)\cdot \ell$ & Yes & $1 \leq k \leq n - 1$ \\
\hline
\end{tabular}
\label{tab:comparison}
\end{table}
\egroup

Some converse results on the sub-packetization level that is necessary for an MDS code to attain the cut-set bound are presented in \cite{GTC14, TWB14}. For $d = n - 1$, Goparaju et al.~\cite{GTC14} show that an exact-repairable MDS code that downloads the same number of symbols from each of the contacted code blocks and employs linear repair schemes satisfies the following bound on its sub-packetization level. 
\begin{align}
\label{eq:gtc_bound}
k \leq 2(\log_{2}\ell)\big(\floorbb{\log_{\frac{n-k}{n-k-1}}\ell} + 1\big).
\end{align}
On the other hand, Tamo et al.~\cite{TWB14} obtain the following lower bound on the sub-packetization level of an MDS vector code which enables exact repair using repair-by-transfer schemes.
\begin{align}
\label{eq:twb_bound}
\ell \geq (n - k)^{\frac{k}{n - k}}.
\end{align}
Note that repair-by-transfer schemes constitute a sub-class of all possible linear repair schemes. In light of the bound in \eqref{eq:twb_bound}, the MDS codes obtained in \cite{SAK15, YeB16b} enable repair-by-transfer mechanisms with optimal repair bandwidth and near-optimal sub-packetization level. However, this sub-packetization level can be prohibitively large for some storage systems, especially when the code has high rate or equivalently has small value of $r = n  - k$. This motivates us to explore the question of designing MDS codes that work with small sub-packetization level and provide repair-by-transfer mechanism for exact repair problem without incurring much degradation in terms of the repair bandwidth. In Table~\ref{tab:comparison}, we compare the proposed construction with the previously known constructions.

The problem of constructing exact-repairable MDS codes with small repair bandwidth and small sub-packetization level has been previously addressed in \cite{zigzag13, RSR13, TK16}. We note that our construction shares some similarities with the constructions presented in \cite{zigzag13, RSR13} as these constructions are obtained by introducing coupling among multiple independent codes as well. However, we work with the parity check matrix view (as opposed to the generator matrix view considered in \cite{zigzag13, RSR13}) which ensures identical repair guarantees for all code blocks without distinguishing between systematic and parity blocks. For $r = n - k = 2$, Tamo and Efremenko construct exact-repairable MDS codes with near-optimal repair bandwidth and sub-packetization level $\ell = O(\log n)$ in \cite{TK16}. The code construction obtained in \cite{TK16} also satisfies the additional requirement that the same amount of data is downloaded from all the contacted blocks.

\noindent \textbf{Exact repair of known codes with small repair bandwidth.} The problem of devising exact repair mechanism with small repair bandwidth for known MDS codes has been studied in \cite{WDB10, SPDG14, XCode14, GW15}. In particular, \cite{SPDG14, GW15} consider the exact repair problem for the well-known Reed-Solomon codes. In \cite{GW15}, Guruswami and Wootters characterize optimal repair bandwidth for these codes in certain regimes of system parameters. \\

\noindent \textbf{Locally repairable codes.}~Another line of work in distributed storage focuses on {\em locality}, the number of the code blocks contacted during the repair of a single code block,  as a metric to characterize the efficiency of the repair process. The bounds on the failure tolerance of locally repairable codes, the codes with small locality, have been obtained in \cite{Gopalan12, PapDim12, KPLK12, RKSV12} and references therein. Furthermore, the constructions of locally repairable codes that are optimal with respect to these bounds are presented in \cite{Gopalan12, PapDim12, KPLK12, RKSV12, BlaumHH13, TamoBarg14, Gopalan14}. Locally repairable codes that also minimize the repair bandwidth for repair of a code block are considered in \cite{KPLK12, RKSV12}. Here we note that the locally repairable codes are not MDS codes, and thus have extra storage overhead.


\section{Exact-repairable MDS codes with near-optimal repair bandwidth}
\label{sec:main}

This paper aims to construct exact-repairable MDS vector codes with small sub-packetization level $\ell$ while achieving near-optimal repair bandwidth, i.e., incurring a small (multiplicative) loss in terms of the repair bandwidth. Towards this, we first introduce the notion of near-optimal repair bandwidth for MDS codes. 
\begin{definition}
\label{def:ApproxMSR}
Let $\Cc$ be an $[n, k\ell, d_{\min} = n - k + 1, \ell]_{|\B|}$ MDS vector code. We call $\Cc$ to be an $(a, \ell, d)$-exact-repairable MDS code if for every $i \in [n]$ and $\cv = (\cv_1, \cv_2,\ldots, \cv_n) \in \Cc$, we can perform exact repair of the code block $\cv_i$ by contacting $d$ other code blocks and downloading at most 
$
a \big(\frac{d}{d - k + 1}\big)\cdot \ell
$
symbols (over $\B$) from the contacted code blocks.
\end{definition}

\begin{remark}
It follows from the bound in \eqref{eq:cut-set-gen} that for any exact-repairable MDS vector code we must have $a \geq 1$. Thus, $(a = 1, \ell, d)$-exact-repairable MDS codes correspond to the exact-repairable MDS vector codes with optimal repair bandwidth. Moreover, we say an MDS code has near-optimal repair bandwidth if it is an $(a, \ell, d)$-exact-repairable MDS code for a small constant $a$.
\end{remark}

In this paper, we focus on the setting with $d = n - 1$, i.e., all the remaining $n-1$ code blocks are contacted to repair a single code block. We now state the main result of this paper which summarizes the parameters of the codes constructed in this paper.

\begin{theorem}
\label{thm:main}
For an integer $1 \leq t \leq \ceilbb{{n}/{(n-k)}} - 1$ and a suitably chosen large enough field $\B$, the general construction presented in this paper gives $\big(1 + {1}/{t}, (n-k)^t, n-1\big)$-exact-repairable MDS vector codes. Moreover, the obtained codes allow for repair-by-transfer schemes.
\end{theorem}

We present our construction for $t =1$, which gives $(2, n-k, n-1)$-exact-repairable MDS vector codes in Section~\ref{sec:twiceRB}. This construction conveys the main ideas behind our approach and establishes Theorem~\ref{thm:main} for $t = 1$. The general construction which establishes Theorem~\ref{thm:main} for all values of $t$ is presented in Section~\ref{sec:RB_t}.

\begin{remark}
\label{rem:t_opt}
We note that for a given value of $t$, $\big(1 + {1}/{t}\big)$ only serves as a clean upper bound on the repair bandwidth of the codes obtained in this paper. Specifically, if we substitute $t = \ceilbb{{n}/{(n-k)}}$ in general construction (cf.~Section~\ref{sec:RB_t}), we obtain $\big(1, (n - k)^{\ceilbb{{n}/{(n-k)}}}, n-1\big)$-exact-repairable MDS vector codes. This matches the best know sub-packetization level for optimal repair bandwidth, which is also near-optimal by \eqref{eq:twb_bound}. In fact, in this case our construction specializes to the construction from \cite{SAK15}.
\end{remark}


\section{Construction of $(2, n-k, n - 1)$-exact-repairable MDS code}
\label{sec:twiceRB}

In this section, we present a construction of exact-repairable MDS vector codes for all values of $n$ and $k$. These codes have sub-packetization level $\ell = n-k$ and require $d = n-1$ code blocks during the repair process and have their repair bandwidth at most
$
2\big(\frac{n-1}{n - k}\big)\cdot \ell,
$
which is twice the cut-set bound (cf.~\eqref{eq:cut-set}). The construction is described in Section~\ref{subsec:construction2}. We illustrate the repair-by-transfer scheme for the obtained codes in Section~\ref{subsec:repair2}. We argue the MDS property for the construction in Section~\ref{subsec:mdsness2}.

\subsection{Code construction}
\label{subsec:construction2}
For an integer $a > 0$, we use $[a]$ to denote the set $\{1,2,\ldots, a\}$. Let $r = n - k$. For ease of exposition, we assume that $r | n$ and $n = sr$. We partition the $n$ code blocks in $r = n-k$ groups of size $s$ each\footnote{For a setting where $r\nmid n$, we can partition the $n$ code blocks in $r = n - k$ groups, $n~({\rm mod}~r)$ groups with $\ceilbb{\frac{n}{r}}$ code blocks and the remaining groups with $\floorbb{\frac{n}{r}}$ code blocks. The rest of the construction can be easily modified to work in this case as well.}. This partitioning allows us to index each code block by a tuple $(u, v)$ where $u \in [r] = [n-k]$ and $v \in [s]$. In particular, for $i \in [n]$ the associated tuple $(u, v)$ satisfies
$
i = (u-1)s + v.
$
With this notation in place, for $(u, v) \in [r] \times [s]$, we denote the $\big((u-1)s + v\big)$-th code block as 
$$\cv_{(u-1)s + v} = \cv_{(u, v)} = \big(c(1; (u, v)), c(2; (u,v)),\ldots, c(r; (u, v)) \big) \in \B^{r}.$$
Here, for $x \in [r]$, $c(x; (u, v))$ denotes the $x$-th symbol (over $\B$) of the $\big((u-1)s + v\big)$-th code block. 

In order to construct an $[n, k\ell, d_{\min} = n- k + 1, \ell = r = n-k]_{\B}$ MDS vector code $\Cc$, we specify an $r\ell \times n\ell$ (or $r^2 \times nr$ for our choice of $\ell$) parity-check matrix $\Pm$ for the code $\Cc$. We classify the linear constraints defined by the parity-check matrix $\Pm$ into two types. Let $\{\lambda_i\}_{i \in [n]}$ be $n$ distinct non-zero elements of $\B$.
\begin{itemize}
\item {\bf {\rm Type~I} constraints:}~We have $r$ {\rm Type~I} constraints which are defined by the first $r$ rows of the matrix $\Pm$. For every $x \in [r]$, we have
\begin{align}
\label{eq:type1}
\sum_{(u,v) \in [r]\times [s]}c(x; (u,v)) = 0.
\end{align}
In Example~\ref{ex:ex1} below, the {\rm Type~I} constraints correspond to the identity blocks of the matrix $\Pm$ (cf.~\eqref{eq:Pm}).
\item {\bf {\rm Type~II} constraints:}~Let $\rho$ be an indeterminate which we specify later. We have $(r-1)\ell = (r-1)r$ {\rm Type~II} constraints which are defined as follows. For every $p \in \{1,\ldots, r - 1\}$ and $x \in [r]$, we have 
\begin{align}
\label{eq:type2}
\underbrace{\sum_{(u, v) \in [r]\times [s]}\lambda^p_{(u-1)s + v} \cdot c(x; (u,v))}_{\text{(a)}} + \underbrace{\sum_{v \in [s]}\rho \cdot c(\overline{x + p}; (x, v))}_{\text{(b)}} = 0,
\end{align}
where for a strictly positive integer $l$, the quantity $\overline{l}$ is defined as follows.
\begin{align}
\label{eq:mod_def}
\overline{l} = \begin{cases}
r &\mbox{if}~l~({\rm mod}~r) = 0, \\
l~({\rm mod}~r) &\mbox{otherwise}.
\end{cases}
\end{align}
We can partition the {\rm Type~II} constraints (cf.~\eqref{eq:type2}) into $(r - 1)$ groups (each group containing $\ell = r$ linear constraints) according to the value of $p \in \{1,\ldots, r-1\}$. In particular, $r$ constraints associated with the same value of $p$ constitute those $r$ rows of the parity-check matrix $\Pm$ which are indexed by the set $\{pr + 1,\ldots, (p+1)r\}$. (See the non-identity blocks of the matrix $\Pm$ in \eqref{eq:Pm}.)
\end{itemize}
\begin{example}
\label{ex:ex1}
We illustrate the construction with an example. Assume that $n = 6$ and $k = 3$, i.e., $n - k = r = 3$. For these values of the system parameters, our $9 \times 18$ parity check matrix takes the following form. 
\begin{align}
\label{eq:Pm}
\footnotesize
\Pm &= \left(\begin{array}{ccc|ccc?ccc|ccc ? ccc|ccc}
1 & 0 & 0 &1 & 0 & 0 &1 & 0 & 0 &1 & 0 & 0 &1 & 0 & 0 &1 & 0 & 0  \\ 
0 & 1 & 0 &0 & 1 & 0 &0 & 1 & 0 &0 & 1 & 0 &0 & 1 & 0 &0 & 1 & 0  \\ 
0 & 0 & 1 &0 & 0 & 1&0 & 0 & 1 &0 & 0 & 1 &0 & 0 & 1 &0 & 0 & 1  \\  \hline
\lambda_1 & {\color{red} \rho} & 0 & \lambda_2 & {\color{red} \rho} & 0 &\lambda_3 & 0 & 0 &\lambda_4 & 0 & 0 &\lambda_5 & 0 & 0 &\lambda_6 & 0 & 0  \\ 
0 & \lambda_1 & 0 & 0 & \lambda_2 & 0 & 0 &\lambda_3 & {\color{red} \rho} & 0 &\lambda_4 & {\color{red} \rho} & 0 &\lambda_5 & 0 & 0 &\lambda_6 & 0  \\  
0 & 0 & \lambda_1 & 0 & 0 & \lambda_2 & 0 & 0 &\lambda_3 & 0 & 0 &\lambda_4 & {\color{red} \rho} & 0 &\lambda_5 & {\color{red} \rho} & 0 &\lambda_6 \\  \hline
\lambda^2_1 & 0 & {\color{red} \rho} & \lambda^2_2 & 0 & {\color{red} \rho} &\lambda^2_3 & 0 & 0 &\lambda^2_4 & 0 & 0 &\lambda^2_5 & 0 & 0 &\lambda^2_6 & 0 & 0  \\ 
0& \lambda^2_1 & 0 & 0 & \lambda^2_2 & 0 & {\color{red} \rho} &\lambda^2_3 & 0 & {\color{red} \rho} &\lambda^2_4 & 0 & 0 &\lambda^2_5 & 0 & 0 &\lambda^2_6 & 0 \\ 
0 & 0 & \lambda^2_1 & 0 & 0 & \lambda^2_2 & 0 & 0 &\lambda^2_3 & 0 & 0 &\lambda^2_4 & 0 & {\color{red} \rho} &\lambda^2_5 & 0 & {\color{red} \rho} &\lambda^2_6  \\ 
\end{array} \right).
\end{align}
The matrix $\Pm$ can be viewed as the perturbation of the block matrix  $\Hm$ which is obtained by replacing all $\rho$ entries in $\Pm$ with zeros. In particular, we can rewrite the matrix $\Pm$ as 
$$
\Pm = \Hm + \Em^{\rho},
$$
where $\Em^{\rho}$ denotes the $9 \times 18$ matrix which contains all the $\rho$ entries in $\Pm$ (cf.~\eqref{eq:Pm}) as its only non-zero entries. (See Figure~\ref{fig:PHE}.) Note that the block matrix $\Hm$ (with diagonal blocks) is a parity-check matrix of an $[n = 6, k\ell = 9, d_{\min} = 4, \ell = 3]_{\B}$ MDS vector code. Here, we also point out that the matrix $\Hm$ is defined by {\rm Type~I} constraints  (cf.~\eqref{eq:type1}) and the part $(a)$ of the {\rm Type~II} constraints (cf.~\eqref{eq:type2}). Similarly, the perturbation matrix $\Em^{\rho}$ is defined by the part $(b)$ of the {\rm Type~II} constraints (cf.~\eqref{eq:type2}).
\end{example}

\begin{figure}[htbp]
\begin{align}
\tiny
\Hm &= \left(\begin{array}{ccc|ccc|ccc|ccc|ccc|ccc}
1 & 0 & 0 &1 & 0 & 0 &1 & 0 & 0 &1 & 0 & 0 &1 & 0 & 0 &1 & 0 & 0  \\ 
0 & 1 & 0 &0 & 1 & 0 &0 & 1 & 0 &0 & 1 & 0 &0 & 1 & 0 &0 & 1 & 0  \\ 
0 & 0 & 1 &0 & 0 & 1&0 & 0 & 1 &0 & 0 & 1 &0 & 0 & 1 &0 & 0 & 1  \\  \hline
\lambda_1 & 0 & 0 & \lambda_2 & 0 & 0 &\lambda_3 & 0 & 0 &\lambda_4 & 0 & 0 &\lambda_5 & 0 & 0 &\lambda_6 & 0 & 0  \\ 
0 & \lambda_1 & 0 & 0 & \lambda_2 & 0 & 0 &\lambda_3 & 0 & 0 &\lambda_4 & 0 & 0 &\lambda_5 & 0 & 0 &\lambda_6 & 0  \\  
0 & 0 & \lambda_1 & 0 & 0 & \lambda_2 & 0 & 0 &\lambda_3 & 0 & 0 &\lambda_4 & 0 & 0 &\lambda_5 & 0 & 0 &\lambda_6 \\  \hline
\lambda^2_1 & 0 & 0 & \lambda^2_2 & 0 & 0 &\lambda^2_3 & 0 & 0 &\lambda^2_4 & 0 & 0 &\lambda^2_5 & 0 & 0 &\lambda^2_6 & 0 & 0  \\ 
0& \lambda^2_1 & 0 & 0 & \lambda^2_2 & 0 & 0 &\lambda^2_3 & 0 & 0 &\lambda^2_4 & 0 & 0 &\lambda^2_5 & 0 & 0 &\lambda^2_6 & 0 \\ 
0 & 0 & \lambda^2_1 & 0 & 0 & \lambda^2_2 & 0 & 0 &\lambda^2_3 & 0 & 0 &\lambda^2_4 & 0 & 0 &\lambda^2_5 & 0 & 0 &\lambda^2_6  \\
\end{array} \right)  \nonumber
\end{align} 
\begin{align}
\tiny
\Em^{\rho} &= \left(\begin{array}{ccc|ccc|ccc|ccc|ccc|ccc}
0 & 0 & 0 &0 & 0 & 0 & 0 & 0 & 0 & 0 & 0 & 0 & 0 & 0 & 0 &0 & 0 & 0  \\ 
0 & 0 & 0 &0 & 0 & 0 &0 & 0 & 0 &0 & 0 & 0 &0 & 0 & 0 &0 & 0 & 0  \\ 
0 & 0 & 0 &0 & 0 & 0&0 & 0 & 0 &0 & 0 & 0 &0 & 0 & 0 &0 & 0 & 0  \\  \hline
0 & {\color{red} \rho} & 0 & 0 & {\color{red} \rho} & 0 & 0 & 0 & 0 & 0 & 0 & 0 &0 & 0 & 0 & 0 & 0 & 0  \\ 
0 & 0 & 0 & 0 & 0 & 0 & 0 & 0 & {\color{red} \rho} & 0 & 0 & {\color{red} \rho} & 0 & 0 & 0 & 0 & 0 & 0  \\  
0 & 0 & 0 & 0 & 0 & 0 & 0 & 0 & 0 & 0 & 0 & 0 & {\color{red} \rho} & 0 & 0 & {\color{red} \rho} & 0 & 0 \\  \hline
0 & 0 & {\color{red} \rho} & 0 & 0 & {\color{red} \rho} & 0 & 0 & 0 & 0 & 0 & 0 & 0 & 0 & 0 & 0 & 0 & 0  \\ 
0& 0 & 0 & 0 & 0 & 0 & {\color{red} \rho} & 0 & 0 & {\color{red} \rho} & 0 & 0 & 0 & 0 & 0 & 0 & 0 & 0 \\ 
0 & 0 & 0 & 0 & 0 & 0 & 0 & 0 & 0 & 0 & 0 & 0 & 0 & {\color{red} \rho} & 0 & 0 & {\color{red} \rho} & 0 \\ 
\end{array} \right)  \nonumber
\end{align}
\caption{Illustration of matrices $\Hm$ and $\Em^{\rho}$ in Example~\ref{ex:ex1}.} 
\label{fig:PHE}
\end{figure}

\subsection{Exact repair of a code block}
\label{subsec:repair2}

Let $(u^{\ast}, v^{\ast}) \in [r] \times [s]$ be the tuple associated with the code block to be repaired. Note that we need to reconstruct the $r$ symbols
$\big\{c(1;(u^{\ast}, v^{\ast})), c(2;(u^{\ast}, v^{\ast})),\ldots, c(r; (u^{\ast}, v^{\ast}))\big\}.$
We divide the repair process in the following two stages.
\begin{enumerate}
\item First, we recover the symbol $c(u^{\ast}; (u^{\ast}, v^{\ast}))$ using the {\rm Type~I} constraint containing it (cf.~\eqref{eq:type1}), i.e., 
\begin{align}
\label{eq:type1a2}
\sum_{(u,v) \in [r]\times [s]}c(u^{\ast}; (u,v)) = 0.
\end{align}
We download the $n-1$ symbols $$\big\{c(u^{\ast}; (u, v))~:~(u, v) \in [r] \times [s]~\text{s.t.}~(u, v) \neq (u^{\ast}, v^{\ast})\big\}$$ from the remaining $d = n - 1$ code blocks in this stage.
\item Next, we sequentially recover the $r-1$ symbols $$\{c(1;(u^{\ast}, v^{\ast})),\ldots, c(u^{\ast}-1;(u^{\ast}, v^{\ast})), c(u^{\ast}+1;(u^{\ast}, v^{\ast})),\ldots, c(r; (u^{\ast}, v^{\ast}))\}$$ using the following $r-1$ {\rm Type~II} constraints (cf.~\eqref{eq:type2}).
\begin{align}
\label{eq:type2a2}
\underbrace{\sum_{({u}, {v}) \in [r]\times [s]}\lambda^p_{({u}-1)s + {v}} \cdot c(u^{\ast}; ({u},{v}))}_{\text{(a)}} + \underbrace{\sum_{{v} \in [s]}\rho \cdot c(\overline{u^{\ast}+p}; (u^{\ast}, {v}))}_{\text{(b)}} = 0~~\text{for}~p \in \{1,\ldots, r-1\}.
\end{align}
Note that the choice of {\rm Type~I} constraint used in the previous stage ensures that we now know all the values of the linear combinations in the part (a) of these {\rm Type~II} linear constraints. Now assuming that $p \in \{1,\ldots, r-1\}$ is such that $\overline{u^{\ast} + p} = \hat{u} \in [r] \backslash \{u^{\ast}\}$, by downloading the additional $s - 1 = \frac{n}{r} - 1$ symbols
$
\big\{c(\hat{u};(u^{\ast}, {v})~:~{v} \in [s]~\text{s.t.}~{v} \neq v^{\ast}\big\}
$
which appear in the part $(b)$ of the linear constraint associated with the underlying value of $p$, we can recover the desired symbol $c(\hat{u};(u^{\ast}, v^{\ast}))$. Thus, the entire second stage involves downloading the following number of symbols (in addition to the symbols downloaded in the first stage).
\begin{align}
(r-1)(s-1) = (r-1)\left({n}/{r} - 1\right) \leq r \left({n}/{r} - 1\right) = n - r \leq n - 1. \nonumber
\end{align}
\end{enumerate}
Note that the entire repair-by-transfer scheme described above downloads at most 
$
2(n-1) = 2\left(\frac{n-1}{n-k}\right)\cdot \ell
$
symbols (over $\B$), which is twice the cut-set bound (cf.~\eqref{eq:cut-set}). 

\subsection{MDS property of the proposed codes}
\label{subsec:mdsness2}

Next, we argue that the construction proposed in Section~\ref{subsec:construction2} gives us MDS vector codes. Let $n$ and $k$ be given system parameters. We show a way to choose the field $\B$ and assign a value to the indeterminate $\rho \in \B$~(cf.~\eqref{eq:type2}) so that the parity-check matrix $\Pm$ defining the obtained code corresponds to a parity-check matrix of an $[n, k\ell = k(n-k), d_{\min} = n - k + 1, \ell = n-k]_{\B}$ MDS vector code. This is equivalent to showing that for every $\Sc \subseteq [n]$ such that $|\Sc| = r = n-k$, the $r\ell \times r\ell$ sub-matrix $\Pm_{\Sc}$ (cf.~\eqref{eq:ParitySub}) is full rank. As illustrated in Example~\ref{ex:ex1}, the matrix $\Pm$ is obtained by perturbing a parity check matrix of an MDS vector code. In particular, we have 
\begin{align}
\label{eq:Pm_mds}
\Pm = \Hm + \Em^{\rho},
\end{align}
where $\Hm$ is a parity check matrix of an $[n, k\ell = k(n-k), d_{\min} = n - k + 1, \ell = n-k]_{\B}$ MDS vector code. Assuming that we can find an irreducible polynomial of large enough degree over a field $\Lf$ with $|\Lf| \geq n+1$, it follows from Proposition~\ref{prop:perturbedMDS} that one can select a field $\B$ and a non-zero element $\rho \in \B$ such that the matrix $\Pm$ corresponds to a parity-check matrix of an MDS vector code. 

\begin{proposition}
\label{prop:perturbedMDS}
Let $\Lf$ be a field with at least $n+1$ elements and $\{\lambda_i\}_{i \in [n]}$ (cf.~Section~\ref{subsec:construction2}) be $n$ distinct non-zero elements in the field $\Lf$. Assume that we can find an irreducible polynomial in $\Lf[X]$ of large enough degree. Then, one can construct a field $\B$ and select a non-zero element $\rho \in \B$ (cf.~\eqref{eq:type2}) such that the matrix $\Pm = \Hm + \Em^{\rho}$ (cf.~\eqref{eq:Pm_mds}) is a parity-check matrix of an $[n, k\ell, d_{\min} = n - k + 1, \ell]_{|\B|}$ MDS vector code.
\end{proposition}
\begin{proof}
See Appendix~\ref{appen:perturbedMDS}.
\end{proof}

The choice of $\rho$ and $\B$ presented above gives us a fully explicit construction of the exact-repairable MDS vector codes (cf.~Appendix~\ref{appen:perturbedMDS}). However, this approach requires the size of the field $\B$ to be quite large. In particular, we need to have $|\B| \gg n^{(r - 1)\ell}$. Next, we highlight another approach which ensures the existence of a choice for $\rho$ such that the code obtained from the proposed construction (cf.~Section~\ref{subsec:construction2}) is an MDS vector code.

\subsubsection{Random perturbations of a parity-check matrix of an MDS vector code}

Note that the matrix $\Hm$ (cf.~\eqref{eq:Pm_mds}) is a parity check matrix of an $[n, k\ell, d_{\min} = n - k + 1, \ell = n-k]_{\B}$ code. If we randomly assign $\rho$ to a non-zero element in $\B$, the parity check matrix $\Pm$ of the obtained code is a random perturbation of the matrix $\Hm$. Assuming that the field $\B$ has large enough size, it follows from Proposition~\ref{prop:perturbedMDS_random} presented below that there exists a choice for the indeterminate $\rho$ such that the matrix $\Pm$ corresponds to a parity-check matrix of an MDS vector code. This approach requires $|\B| \gg n^{r}r\ell$. We note that even though this alternative approach requires a slightly smaller field, it does not give us a fully explicit construction. Here, we also point out that Proposition~\ref{prop:perturbedMDS_random} follows from the analysis presented in \cite{SAK15}.
\begin{proposition}
\label{prop:perturbedMDS_random}
Assume that $\B$ is a field of large enough size and $\Hm \in \B^{r\ell \times n\ell}$ is a parity-check matrix of an $[n, k\ell, d_{\min} = n - k + 1, \ell]_{|\B|}$ MDS vector code. Let $\Em^{\rho}$ be a random $r\ell \times n\ell$ matrix with all of its non-zero entries equal to an element $\rho$ which is selected uniformly at random from the non-zero elements in $\B$. Then, the probability that the matrix $\Pm = \Hm + \Em^{\rho}$ is a parity-check matrix of an $[n, k\ell, d_{\min} = n - k + 1, \ell]_{|\B|}$ MDS vector code is bounded away from zero.
\end{proposition}

\begin{proof}
See Appendix~\ref{appen:perturbedMDS_random}.
\end{proof}


\section{Construction of $(1 + \frac{1}{t}, (n-k)^t, n - 1)$-exact-repairable MDS vector codes}
\label{sec:RB_t}

In this section, we generalize the construction presented in Section~\ref{sec:twiceRB}. A design parameter $t$ allows us to increase the sub-packetization level $\ell$ in order to decrease the repair bandwidth of the code. In particular, for the integer $1 \leq t \leq \ceilb{{n}/{r}} - 1 = \ceilb{{n}/{(n-k)}} - 1$, we design exact-repairable MDS vector codes with sub-packetization level $\ell = r^t = (n-k)^t$, $d = n - 1$ and repair bandwidth at most 
$$
\left(1 + \frac{1}{t}\right)\left(\frac{n-1}{n - k}\right)\cdot \ell~~\text{symbols (over $\B$)}.
$$ 
This repair bandwidth is at most $\big(1 + \frac{1}{t}\big)$ times the cut-set bound (cf.~\eqref{eq:cut-set}). Recall that, for an integer $a > 0$, we use $[a]$ to denote the set $\{1,2,\ldots, a\}$. 

\subsection{Code construction}
\label{subsec:construction_t}
Similar to Section~\ref{subsec:construction2}, for ease of exposition, we assume that $r | n$ and $n = sr = s(n-k)$. We partition the $n$ code blocks in $r = n-k$ groups of equal sizes with each group containing $s = \frac{n}{r} = \frac{n}{n-k}$ code blocks. Using this partition, we index each code block by a tuple $(u, v)$ where $u \in [r] = [n-k]$ and $v \in [s]$. In particular, for $i \in [n]$ the associated tuple $(u, v)$ satisfies
$$
i = (u-1)s + v.
$$
Furthermore, we index the $r^{t} = (n-k)^t$ symbols (over $\B$) in each code block by the $r^t$ distinct $t$-length vectors in $[r]^t = [n-k]^t$. For $(u, v) \in [r] \times [s]$, the $\big((u-1)s + v\big)$-th code block can be represented as follows.
$$
\cv_{(u-1)s + v} = \cv_{(u, v)} = \big\{c\big((x_1, x_2,\ldots, x_t); (u, v)\big) \big\}_{(x_1, x_2,\ldots, x_t) \in [r]^t}.
$$
Let $\{\lambda_i\}_{i \in [n]}$ be $n$ distinct non-zero elements of $\B$. We are now ready to present our construction of an $\big(1 + \frac{1}{t}, \ell = r^t, d = n - 1\big)$-exact-repairable MDS vector code $\Cc$ by defining an $r\ell \times n\ell$ parity-check matrix $\Pm$ of the code $\Cc$. Specifically, we classify the $r\ell = r^{t+1}$ linear constraints defined by the parity-check matrix $\Pm$ into two types. 
\begin{itemize}
\item {\bf {\rm Type~I} constraints:}~We have $\ell = r^t$ {\rm Type~I} constraints which are defined by the first $\ell = r^t$ rows of the matrix $\Pm$. For every $(x_1, x_2,\ldots, x_t) \in [r]^t$, we have
\begin{align}
\label{eq:type1t}
\sum_{(u,v) \in [r]\times [s]}c\big((x_1, x_2,\ldots, x_t); (u, v)\big) = 0.
\end{align}
\item {\bf {\rm Type~II} constraints:}~Let $\rho$ be an indeterminate which we specify later. We have $(r-1)\ell = (r-1)r^t$ {\rm Type~II} constraints. Recall that for strictly positive integers $l$ and $m$, the quantity $\overline{l}^{\{m\}}$ is defined as follows.
\begin{align}
\overline{l}^{\{m\}} = \begin{cases}
m &\mbox{if}~l~({\rm mod}~m) = 0, \\
l~({\rm mod}~m) &\mbox{otherwise}.
\end{cases}
\end{align} 
Assuming that $v \in [s] = \left[\frac{n}{r}\right]$ be such that $\overline{v}^{\{t\}} = a \in [t]$ and $p \in \{1, 2,\ldots, r-1\}$, we use $\overline{(x_1, x_2,\ldots, x_t)}^{v, p}$ to denote the vector obtained by modifying a single coordinate of the vector $(x_1, x_2,\ldots, x_t)$ in the following manner. 
\begin{align}
\overline{(x_1, x_2,\ldots, x_t)}^{v, p} &= (x_1, x_2,\ldots, x_{a-1}, \overline{x_{\overline{v}^{\{t\}}} + p}^{\{r\}},x_{a+1},\ldots, x_t) \nonumber \\
&= (x_1, x_2,\ldots, x_{a-1}, \overline{x_{a} + p}^{\{r\}},x_{a+1},\ldots, x_t). \nonumber 
\end{align}
For every $p \in \{1,\ldots, r - 1\}$ and $(x_1, x_2,\ldots, x_t) \in [r]^t$, we have an associated linear constraint in the parity-check matrix $\Pm$.
\begin{align}
\label{eq:type2t}
&\underbrace{\sum_{(u, v) \in [r]\times [s]}\lambda^p_{(u-1)s + v} \cdot c\big((x_1, x_2,\ldots, x_t); (u, v)\big)}_{\text{(a)}} \nonumber \\
&+ \underbrace{\sum_{v \in [s]}\rho \cdot c\big(\overline{(x_1, x_2,\ldots, x_t)}^{v, p}; (x_{\overline{v}^{\{t\}}}, v)\big)}_{\text{(b)}} = 0.
\end{align}
We can partition the {\rm Type~II} constraints (cf.~\eqref{eq:type2t}) into $(r - 1)$ groups (each group containing $\ell = r^t$ linear constraints) according to the value of $p \in \{1,\ldots, r-1\}$. In particular, $r^t$ constraints associated with the same value of $p$ constitute those $r^t$ rows of the parity-check matrix $\Pm$ which are indexed by the set $$\{pr^t + 1,\ldots, (p+1)r^t\} \subseteq [r\ell] = [r^{t+1}].$$
\end{itemize}

\begin{example}
\label{ex:type2t}
In this example, we look at the composition of a {\rm Type~II} constraint (cf.~\ref{eq:type2t}) when $t=2$. We assume that $n = 9$ and $r = n - k = 3$. This implies that $s = \frac{n}{n-k} = 3$. For $(x_1, x_2) \in [r]^t =  [3]^2$ and $p = 1$ the associated {\rm Type~II} constraint takes the following form.
\begin{align}
\label{eq:type2t_ex}
&~~~~~~~~~~~~~~~~~~~\underbrace{\sum_{(u, v) \in [3]\times [3]}\lambda_{(u-1)3 + v} \cdot c\big((x_1, x_2); (u, v)\big)}_{\text{(a)}}~~~~~~+ \nonumber \\
& \underbrace{\rho \cdot c\big((\overline{x_1 + 1}^{\{3\}}, x_2); (x_1, 1)\big) + \rho \cdot c\big((x_1, \overline{x_2 + 1}^{\{3\}}); (x_2, 2)\big) + \rho \cdot c\big((\overline{x_1 + 1}^{\{3\}}, x_2); (x_1, 3)\big)}_{\text{(b)}} = 0.
\end{align}
Note that we have used the following equalities in \eqref{eq:type2t_ex} which hold for $t = 2$ and $s = \frac{n}{n-k} = 3$.
\begin{align}
\overline{1}^{\{t = 2\}} = \overline{3}^{\{2\}} = 1~~\text{and}~~\overline{2}^{\{2\}} = 2.
\end{align}
\end{example}

\subsection{Exact repair of failed code blocks in the proposed codes}
\label{subsec:repair_t}

We now illustrate a mechanism to perform exact repair of a code block in the code obtained by the construction proposed in Section~\ref{subsec:construction_t}. Let $(u^{\ast}, v^{\ast}) \in [r] \times [s]$ be the tuple associated with the code block to be repaired. Note that we need to reconstruct the following $r^t$ code symbols.
\begin{align}
\big\{c\big((x_1, x_2,\ldots, x_t);(u^{\ast}, v^{\ast})\big)\big\}_{(x_1, x_2,\ldots, x_t) \in [r]^t}.
\end{align}
Similar to Section~\ref{subsec:repair2}, we divide the repair process in the following two stages.
\begin{enumerate}
\item In the first stage we utilize {\rm Type~I} constraints (cf.~\eqref{eq:type1t}) to recover the following $r^{t-1}$ symbols.
\begin{align}
\label{eq:stage1t}
\big\{c\big((x_1,\ldots,x_{a-1}, x_{a} = u^{\ast}, x_{a+1},\ldots, x_t);(u^{\ast}, v^{\ast})\big)\big\}_{(x_1,\ldots, x_{a-1}, x_{a+1},\ldots, x_t) \in [r]^{t-1}},
\end{align} 
where $a = \overline{v^{\ast}}^{\{t\}}$. Recall that. for $(x_1, x_2,\ldots, x_{t}) \in [r]^t$, the {\rm Type~I} constraint takes the following form.
\begin{align}
\label{eq:type1a}
\sum_{(u,v) \in [r]\times [s]}c\big((x_1, x_2,\ldots, x_t)); (u,v)\big) = 0.
\end{align}
Therefore, in order to recover the $r^{t-1}$ symbols shown in \eqref{eq:stage1t} using these constraints, we download the following $(n-1)r^{t-1}$ symbols from the remaining $n-1$ code blocks. 
\begin{align}
\label{eq:repair_stage1}
&\big\{c\left((x_1,\ldots,x_{a-1}, x_{a} = u^{\ast}, x_{a+1},\ldots, x_t); ({u}, {v})\right)~:~(x_1,\ldots, x_{a-1}, x_{a+1},\ldots, x_t) \in [r]^{t-1}\nonumber \\
&~~~~~~~~~~~~~~~~~~~~~~~~~~~~~~~~~~~~~~~~~~~~~~~~~~~~~~~~~~~~~~~~~~~~~~~~~~~~~~~\text{and}~({u}, {v}) \in [r] \times [s]~\text{s.t.}~({u}, {v}) \neq (u^{\ast}, v^{\ast})\big\},
\end{align}
where $a = \overline{v^{\ast}}^{\{t\}}$.
\item Note that, at the end of the stage $1$ of the repair process, we have access to the following symbols (over $\B$) which also include the $r^{t-1}$ symbols recovered in the stage $1$. 
\begin{align}
\label{eq:repair_stage1}
&\big\{c\left((x_1,\ldots,x_{a-1}, x_{a} = u^{\ast}, x_{a+1},\ldots, x_t); (u, v)\right)~:~(x_1,\ldots, x_{a-1}, x_{a+1},\ldots, x_t) \in [r]^{t-1}\nonumber \\
&~~~~~~~~~~~~~~~~~~~~~~~~~~~~~~~~~~~~~~~~~~~~~~~~~~~~~~~~~~~~~~~~~~~~~~~~~~~~~~~~~~~~~~~~~~~~~~~~~~~~~~~~~~~~~~~~\text{and}~({u}, {v}) \in [r] \times [s]\big\}.
\end{align}
In the stage $2$ of the repair process, we employ the {\rm Type~II} constraints to sequentially recover the remaining $(r-1)r^{t-1}$ symbols
\begin{align}
\label{eq:stage2t}
\big\{c\big((x_1,\ldots,x_{a-1}, x_{a} \neq u^{\ast}, x_{a+1},\ldots, x_t);(u^{\ast}, v^{\ast})\big)\big\}_{(x_1,\ldots, x_{a-1}, x_{a+1},\ldots, x_t) \in [r]^{t-1}},
\end{align}
where $a = \overline{v^{\ast}}^{\{t\}}$. Let $p \in \{1,2,\ldots, r-1\}$ be such that we have $\overline{u^{\ast} + p}^{\{r\}} = \hat{u} \in [r]\backslash\{u^{\ast}\}$. We utilize the following {\rm Type~II} constraint to repair the desired symbol $c\big((x_1,\ldots, x_{a-1}, \hat{u}, x_{a+1},\ldots, x_t);(u^{\ast}, v^{\ast})\big)$.
\begin{align}
\label{eq:type2at}
&\underbrace{\sum_{({u}, {v}) \in [r]\times [s]}\lambda^p_{(u-1)s + v} \cdot c\big((x_1,\ldots, x_{a-1}, x_{a} = u^{\ast}, x_{a+1},\ldots, x_t); ({u}, {v})\big)}_{\text{(a)}} + \nonumber \\
& \underbrace{\rho \cdot c\big((x_1,\ldots, x_{a-1},\hat{u}, x_{a+1},\ldots, x_t); (x_{a} = u^{\ast}, v^{\ast})\big)}_{\text{(b-I)}} + \nonumber \\
& \underbrace{\sum_{v \in [s]~:~v \neq v^{\ast}}\rho \cdot c\big(\overline{(x_1,\ldots, x_{a} = u^{\ast},\ldots, x_t)}^{v, p}; (x_{\overline{v}^{\{t\}}}, v)\big)}_{\text{(b-II)}} = 0.
\end{align}
It is straightforward to verify that at the end of the stage $1$ of the repair process, we know the value of the linear combination in the part (a) of this linear constraint (cf.~\eqref{eq:repair_stage1}). We now argue that we also know many of the symbols appearing in the part (b-II) of this constraint. Note that the part (b-II) can be rewritten as follows.
\begin{align}
\label{eq:type2at_bII}
&\sum_{v \in [s]~:~v \neq v^{\ast}}\rho \cdot c\big(\overline{(x_1,\ldots, x_{a} = u^{\ast},\ldots, x_t)}^{v, p}; (x_{\overline{v}^{\{t\}}}, v)\big) \nonumber \\
& = \underbrace{\sum_{v \neq v^{\ast}~:~\overline{v}^{\{t\}} =\overline{v^{\ast}}^{\{t\}} = a}\rho \cdot c\big((x_1,\ldots,x_{a-1},x_{a} =  \hat{u}, x_{a+1},\ldots, x_t); (x_{a} = u^{\ast}, v)\big)}_{\text{(b-II-1)}} \nonumber \\
&~~~~~+ \underbrace{\sum_{v \neq v^{\ast}~:~\overline{v}^{\{t\}} \neq \overline{v^{\ast}}^{\{t\}} = a}\rho \cdot c\big(\overline{(x_1,\ldots, x_{a} = u^{\ast},\ldots, x_t)}^{v, p}; (x_{\overline{v}^{\{t\}}}, v)\big)}_{\text{(b-II-2)}}.
\end{align}
Note that the code symbols appearing in part (b-II-2) are indexed by the vectors which have their $a$-th coordinate equal to $u^{\ast}$. One can verify that these symbols are already known at the end of the stage $1$ of the repair process (cf.~\eqref{eq:repair_stage1}). Therefore, in order to recover the desired symbol $$c\big((x_1,\ldots, x_{a-1}, \hat{u}, x_{a+1},\ldots, x_t);(u^{\ast}, v^{\ast})\big)$$ using the linear constraint in \eqref{eq:type2at}, we need to only download the code symbols appearing in the part (b-II-1). Note that there are at most $\floorb{\frac{s}{t}}$ symbols in the part (b-II-1). Since we have to repair $(r-1)r^{t-1}$ symbols in the stage $2$ (cf.~\eqref{eq:stage2t}), the number of symbols that we download in the stage $2$ (in addition to the symbol downloaded in the stage $1$) is at most
\begin{align}
(r-1)r^{t-1}\floorb{\frac{s}{t}} \leq (r-1)r^{t-1}\left({\frac{s}{t}}\right) \nonumber \\
= \frac{r^{t-1}}{t}\frac{r-1}{r}n \nonumber \\
\overset{(i)}{\leq}\frac{r^{t-1}}{t}(n-1).
\end{align}
Here the step $(i)$ follows as, for $r = n- k \leq n$, we have $\frac{r-1}{r} \leq \frac{n-1}{n}$. Since we download $(n-1)r^{t-1}$ symbols during the stage $1$ of the repair process, the total repair bandwidth is at most 
\begin{align}
(n-1)r^{t-1} + \frac{r^{t-1}}{t}(n-1) = \left(1+ \frac{1}{t}\right)(n-1)r^{t-1} = \left(1+ \frac{1}{t}\right)\left(\frac{n-1}{n-k}\right)\cdot \ell,
\end{align}
which is $\left(1 + 1/t\right)$ times the cut-set bound (cf.~\eqref{eq:cut-set}).
\end{enumerate}

\subsection{MDS property of the proposed codes}
\label{subsec:mds_t}

The argument for this part is identical to that used in Section~\ref{subsec:mdsness2}.

\section{Conclusion and future directions}
\label{sec:conclusion}

We construct MDS vector codes that allow for exact repair of a code block by downloading near-optimal amount of data from the remaining code blocks. These codes are well suited for distributed storage systems as they work with small sub-packetization level and enable repair-by-transfer mechanisms, where repair of a code block (node) requires minimal computation at the contacted code blocks (nodes). We conclude by pointing out a few directions to extend this work.

\begin{itemize}
\item \textbf{Reducing the size of base field $\B$.} The exact-repairability and the corresponding repair bandwidth of the proposed codes only depend on the combinatorial structure, i.e., the locations of non-zero entries, of the designed parity-check matrix. However, the argument which establishes the MDS property for these codes requires the size of the base field $\B$ to be quite large. We note that the similar issue arises in many previous works, e.g.,~\cite{zigzag13, SAK15}. Recently, Ye and Barg have addressed this issue for the codes that operated exactly at the cut-set bound in \cite{YeB16a, YeB16b}. However, they again work with large sub-packetization level $n^{\ceilbb{\frac{n}{n-k}}}$. The reduction of the base field size for our construction is an interesting question, which has both theoretical and practical significance.
\item \textbf{Constructing codes for general values of $d$.} In this  paper we focus on the setting with $d = n - 1$. Extending the construction proposed in this paper for general value of $k < d < n - 1$ is another important direction to explore. Towards this, one relatively straightforward approach is to employ the ideas used in \cite{RKV16a} to extend the construction from \cite{SAK15} to general values of $d$. For an integer $t \geq 1$, this would provide exact-repairable codes with sub-packetization level $(d - k + 1)^t$ and small repair bandwidth. Moreover, the obtained codes would also have repair-by-transfer schemes. 
\item \textbf{Simultaneous repair of multiple code blocks.} The problem of designing MDS codes that allow for simultaneous repair of multiple code blocks has been addressed in several works, including \cite{ShumHu, Kermarrec:Repairing11, RKV16b, YeB16a}. Designing codes that provide mechanisms to perform simultaneous repair of multiple code blocks, and as well as a good trade-off between the sub-packetization level and repair bandwidth is an interesting direction to pursue.
\end{itemize}

\section*{Acknowledgement}
\label{sec:ack}
We would like to thank Itzhak Tamo for introducing us to the problem of constructing exact-repairable MDS codes with near-optimal repair bandwidth during the 2016 Information Theory and Applications (ITA) workshop. We are also grateful to him for commenting on an earlier version of this draft.


\bibliographystyle{plain}
\bibliography{RepairBW_MDS}


\appendix

\section{Necessary sub-packetization level for MDS vector codes}
\label{appen:appen_sub}

Assume that $\F \cong \B^{\ell}$. Let $\Cc \subseteq \F^{n}$ be an MDS vector code with the sub-packetization level $\ell$ where all contacted nodes contribute to the repair process. For a constant $b \geq 1$, let the repair bandwidth of $\Cc$ for exact repair of a single code block is less than $c$ times the cut-set bound, i.e.,
\begin{align}
\text{No. of symbols (over $\B$) downloaded from the contacted $d = n-1$ nodes} \leq  b\left(\frac{n - 1}{n - k}\right)\cdot \ell.
\end{align}
This implies that there exists at least one contacted node which contributes at most $\floorbb{\frac{b\ell}{n - k} }$ symbols (over $\B$) during the repair process. Moreover, each of the contacted $d = n - 1$ nodes sends at least $1$ symbol (over $\B$) during the repair process. Hence, we have that
\begin{align}
\floorb{\frac{b\ell}{n - k}}\geq 1.
\end{align}
This gives us that 
\begin{align}
\ell \geq \frac{n - k}{b}~~\text{or}~~\ell = \Omega(n-k).
\end{align}

\section{Proof of Proposition~\ref{prop:perturbedMDS}.}
\label{appen:perturbedMDS}

Let $\Lf$ be a finite field such that all the $n$ distinct non-zero elements $\{\lambda_i\}_{i \in [n]}$ used in the code construction (cf.~Section~\ref{subsec:construction2}) belong to $\Lf$. Furthermore, let $\rho \notin \Lf$ be an element from an extension of $\Lf$ such that its minimal polynomial $m_{\rho}(x) \in \Lf[X]$ has its degree ${\rm deg}(m_{\rho})$ strictly greater than $(r - 1)\ell$~\cite{Lidl}. We take $\B$ to be $\Lf(\rho)$, the simple extension of the field $\Lf$ to include $\rho$. Recall that $\B = \Lf(\rho) \cong \Lf[X]/\langle m_{\rho}(X)\rangle$, where $\langle m_{\rho}(X)\rangle \subset \Lf[X]$ is the ideal generated by the minimal polynomial $m_{\rho}(X)$. Moreover, $|\B| = |\Lf|^{{\rm deg}(m_{\rho})}$.

Now, we argue that for such a choice of $\rho$ and the associated field $\B$, the parity check matric $\Pm$ (cf.~\eqref{eq:Pm_mds}) define an $[n, k\ell, d_{\min} = n - k + 1, \ell]_{\B}$ MDS vector code, i.e., we have 
\begin{align}
\label{eq:Pm_Sn}
\det(\Pm_{\Sc}) = \det(\Hm_{\Sc} + \Em^{\rho}_{\Sc}) \neq 0~~\forall~~\Sc \subseteq [n]~\text{such that}~|\Sc| = r = n-k.
\end{align}

Recall that $\Hm$ is an $r\ell \times n\ell$ parity-check matrix of an $[n, k\ell, d_{\min} = n - k + 1, \ell]_{|\B|}$ MDS vector code. Therefore, we have 
\begin{align}
\label{eq:HmSn}
\det\big(\Hm_{\Sc}\big) \neq 0~~~\forall~\Sc \subseteq [n]~\text{such that}~|\Sc| = r.
\end{align}

Consider a set $\Sc \subseteq [n]$ such that $|\Sc| = r = n - k$ and the associated sub-matrix $\Hm_{\Sc} + \Em^{X}_{\Sc}$ (cf.~\eqref{eq:Pm_mds}), where $X$ is an indeterminate. Note that the determinant of this $r\ell \times r\ell$ matrix can be expressed as
\begin{align}
\label{eq:f_defn}
f_{\Sc}(X) = \det\big(\Hm_{\Sc} + \Em^{X}_{\Sc} \big),
\end{align}
where $f_{\Sc}(X) \in \Lf[X]$ is a polynomial of degree at most $(r-1)\ell$ and its coefficients are defined by the elements $\{\lambda_{i}\}_{i \in \Sc} \subseteq \Lf$. We now argue that the polynomial $f_{\Sc}(X)$ is a non-trivial (not an identically zero) polynomial. Towards this, we consider the value of the polynomial $f_{\Sc}(X)$ at $X = 0$.
\begin{align}
f_{\Sc}(0) &= \det\big(\Hm_{\Sc} + \Em^{0}_{\Sc}\big) \nonumber \\
&\overset{(i)}{=} \det\big(\Hm_{\Sc}\big)\nonumber \\
&\overset{(ii)}{\neq} 0.
\end{align}
Here the step $(i)$ holds as $\Em^{0}$ reduces to a zero matrix, and the step $(ii)$ follows from \eqref{eq:HmSn}. Since $f_{\Sc}(X)$ evaluates to a non-zero value at $X = 0$, it's a non-trivial polynomial. We now substitute $X = \rho$, which gives us the following (cf.~\ref{eq:f_defn}).
\begin{align}
f_{\Sc}(\rho) &= \det\big(\Hm_{\Sc} + \Em^{\rho}_{\Sc} \big) \nonumber \nonumber  \\
&\overset{(i)}{=} \det\big(\Pm_{\Sc}\big) \nonumber \\
&\overset{(ii)}{\neq} 0. 
\end{align}
Here, the step $(i)$ follows from the definition of $\Pm$ (cf.~\eqref{eq:Pm_mds} and \eqref{eq:Pm_Sn}). The step $(ii)$ follows as we have that the degree of $f_{\Sc}(X) \in \Lf[X]$ is strictly less than the degree of $m_{\rho}(X) \in \Lf[X]$, the minimal polynomial of $\rho$. Since the choice of $\Sc$ is arbitrary over all the subsets of $[n]$ of size $r = n-k$. We have that 
\begin{align}
f_{\Sc}(\rho) = \det\big(\Pm_{\Sc}\big) \neq 0~~\forall~\Sc \subseteq [n]~\text{such that}~|\Sc| = r.
\end{align}
This completes the proof. \qed

\section{Proof of Proposition~\ref{prop:perturbedMDS_random}.}
\label{appen:perturbedMDS_random}

\begin{proof}
Note that  $\Hm \in \B^{r\ell \times n \ell}$ is an $r\ell \times n\ell$ parity-check matrix of an $[n, k\ell, d_{\min} = n - k + 1, \ell]_{|\B|}$ MDS vector code. Therefore, we have 
\begin{align}
\label{eq:HmS}
\det\big(\Hm_{\Sc}\big) \neq 0~~~\forall~\Sc \subseteq [n]~\text{s.t.}~|\Sc| = r.
\end{align}

Consider the perturbed matrix $\Pm = \Hm + \Em^{\nu}$ where $\nu$ denotes an indeterminate. Let $\Sc \subseteq [n]$ such that $|\Sc| = r = n - k$ and $\Pm_{\Sc} = \Hm_{\Sc} + \Em^{\nu}_{\Sc}$ be the associated sub-matrix (cf.~\eqref{eq:ParitySub}). Let $f_{\Sc}(\nu)$ be the determinant of the $r\ell \times r\ell$ matrix $\Pm_{\Sc}$, i.e.,
\begin{align}
\label{eq:f_def}
f_{\Sc}(\nu) = \det\big(\Pm_{\Sc}\big).
\end{align}
Note that $f_{\Sc}(\nu)$ is a polynomial in the indeterminate $\nu$. Next, we argue that $f_{\Sc}(\nu)$ is a non-trivial (not an identically zero) polynomial. Let's consider the value of the polynomial $f_{\Sc}(\nu)$ at $\nu = 0$,
\begin{align}
f_{\Sc}(0) &= \det\big(\Hm_{\Sc} + \Em^{\nu = 0}_{\Sc}\big) \nonumber \\
&\overset{(i)}{=} \det\big(\Hm_{\Sc}\big)\nonumber \\
&\overset{(ii)}{\neq} 0.
\end{align}
Here the step $(i)$ holds as for $\nu = 0$, $\Em^{\nu}$ reduces to a zero matrix. The step $(ii)$ follows from \eqref{eq:HmS}. Since $f_{\Sc}(\nu)$ evaluates to a non-zero value at $\nu = 0$, it's a non-trivial polynomial. Note that this is true for any choice of the set $\Sc \subseteq [n]$ such that $|\Sc| = r = n - k$. We now consider the following polynomial.
\begin{align}
g(\nu) = \prod_{\Sc \subseteq [n]~:~|\Sc| = r} f_{\Sc}(\nu).
\end{align}
Note that $g(\nu)$ is a non-trivial polynomial as it is a product of the non-trivial polynomials $\big\{f_{\Sc}(\nu)\big\}_{\Sc \subseteq [n]~:~|\Sc| = r}$. Moreover, we have that
\begin{align}
{\rm degree}\big(g(\nu)\big) &= \sum_{\Sc \subseteq [n]~:~|\Sc| = r}{\rm degree}\big(f_{\Sc}(\nu)\big) \nonumber \\
&\overset{(i)}{\leq} {n \choose r}r\ell,
\end{align}
where $(i)$ follows from the fact that for every subset $\Sc$ the degree of the associated polynomial $f_{\Sc}(\nu)$ is at most $r\ell$. Now, if we substitute $\nu$ with $\rho$ which is selected uniformly at random from the non-zero elements in $\B^{\ast}$, then if follows from the Schwartz-Zippel lemma that we have 
\begin{align}
\label{eq:det_prob}
\pr{\text{$\Pm = \Hm + \Em^{\rho}$ is not a parity-check matrix of an MDS vector code}} &= \pr{g(a) = 0} \nonumber \\
&\leq \frac{{n \choose r}r\ell}{|\B| - 1}.
\end{align}
Note that, for $|\B|$ large enough (in particular $|\B| =  \Omega\left({n \choose r}r\ell\right)$), the right hand side of \eqref{eq:det_prob} is strictly smaller than $1$.

\end{proof}

\end{document}